\documentclass[12pt,draftclsnofoot,journal,letterpaper,oneside,onecolumn]{IEEEtran}

\usepackage{amssymb}
\usepackage{amsmath}
\usepackage{amsfonts}
\usepackage{graphicx}
\usepackage{subfigure}
\usepackage{cite}
\usepackage{enumerate}
\usepackage{url}

\usepackage{algorithm}
\usepackage{algorithmic}

\begin{document}

\newtheorem{definition}{\bf Definition}
\newtheorem{proposition}{\bf Proposition}
\newtheorem{theorem}{\bf Theorem}
\newtheorem{lemma}{\bf Lemma}

\renewcommand{\algorithmicrequire}{\textbf{Input:}}
\renewcommand{\algorithmicensure}{\textbf{Output:}}

\title{\LARGE{Social Data Offloading in D2D-Enhanced Cellular Networks by Network Formation Games}}

\author{
\IEEEauthorblockN{
\small{Tianyu Wang}\IEEEauthorrefmark{1},
\small{Yue Sun}\IEEEauthorrefmark{1},
\small{Lingyang Song}\IEEEauthorrefmark{1},
\small{and Zhu Han}\IEEEauthorrefmark{2} \\}
\IEEEauthorblockA{
\IEEEauthorrefmark{1}\small{School of Electrical Engineering and Computer Science, Peking University, Beijing, China,\\ Email: tianyu.alex.wang@pku.edu.cn; lingyang.song@pku.edu.cn} \\
\IEEEauthorrefmark{2}\small{Electrical and Computer Engineering Department, University of Houston, Houston, TX, USA, \\ Email: zhan2@uh.edu}\\}
}

\maketitle

\thispagestyle{empty}

\begin{abstract}

Recently, cellular networks are severely overloaded by social-based services, such as YouTube, Facebook and Twitter, in which thousands of clients subscribe a common content provider (e.g., a popular singer) and download his/her content updates all the time. Offloading such traffic through complementary networks, such as a delay tolerant network formed by device-to-device (D2D) communications between mobile subscribers, is a promising solution to reduce the cellular burdens. In the existing solutions, mobile users are assumed to be volunteers who selfishlessly deliver the content to every other user in proximity while moving. However, practical users are selfish and they will evaluate their individual payoffs in the D2D sharing process, which may highly influence the network performance compared to the case of selfishless users. In this paper, we take user selfishness into consideration and propose a \emph{network formation game} to capture the dynamic characteristics of selfish behaviors. In the proposed game, we provide the utility function of each user and specify the conditions under which the subscribers are guaranteed to converge to a stable network. Then, we propose a practical network formation algorithm in which the users can decide their D2D sharing strategies based on their historical records. Simulation results show that user selfishness can highly degrade the efficiency of data offloading, compared with ideal volunteer users. Also, the decrease caused by user selfishness can be highly affected by the cost ratio between the cellular transmission and D2D transmission, the access delays, and mobility patterns.

\end{abstract}

\newpage

%%%%%%%%%%%%%%%%%%%%%%%
\section{Introduction}%
%%%%%%%%%%%%%%%%%%%%%%%

Due to massive growth of smart devices and increasing popularity of mobile applications, mobile traffic is expected to grow continuously at a rapid rate in the next few years. According to Cisco's latest Visual Networking Index (VNI) report, the global mobile data traffic will increase nearly $11$-fold from 2013 to 2018, and the average traffic generated by a smartphone will be $2.7$ GB per month by 2018, a $5$-fold increase over the 2013 average of $529$ MB per month~\cite{Cisco-2014}. In order to cope with this ongoing tsunami of mobile data demand, the operators have put forward multiple solutions, such as upgrading the existing networks, building up new networks, and implementing a usage-based pricing plan.

Mobile traffic offloading, or mobile data offloading, which refers to the use of complementary network technologies (e.g., Wi-Fi networks) to deliver data originally targeted for cellular networks, has attracted a lot of attentions from both the industry and academia~\cite{LLYRC-2013, Cisco-2012, MLQ-2011, PJAK-2013, APJGK-2014, HHKMPS-2010, ICM-2009, LQJHWC-2014, ZGCH-2014}. On the one hand, it is compelling to operators who want to reduce cellular traffic load with a low-cost solution and make better use of their existing assets. On the other hand, it is also appealing to end-users, who enjoy the low price and high data rate of the complementary networks. A detailed survey of mobile traffic offloading can be found in~\cite{AAA-2013}.

Most of the existing efforts have focused on Wi-Fi networks, in which cellular traffic is delivered by Wi-Fi connections when mobile users are within the coverage of Wi-Fi access points~\cite{LLYRC-2013, Cisco-2012}. Femtocells can also be used for mobile traffic offloading, which offload both air interface traffic and core network traffic from cellular systems~\cite{MLQ-2011}. Some other offloading approaches exploit short-range communication techniques, such as device-to-device (D2D) communication and Wi-Fi direct, to offload the traffic of context-aware applications~\cite{PJAK-2013, APJGK-2014}. Recently, social network services (SNSs) have taken a great portion of our daily life. In SNSs, a large number of clients subscribe a common content provider that frequently pushes multimedia content to the subscribers, e.g., text, photos, or videos, which generates thousands of duplicated downloads of the same content, and thus, consumes a great amount of bandwidth in cellular systems~\cite{CKRAM-2007}.

Many efforts are being carried out to offload such traffic by exploiting opportunistic communications among mobile users~\cite{HHKMPS-2010, ICM-2009, LQJHWC-2014, ZGCH-2014}. Specifically, the content provider first pushes the content to a target set of users as seeds via cellular networks, and later, the seed users use short-range communication techniques to share the content with other users in proximity while moving. At last, if some user fails to receive the content after a certain delay, it will download directly from the content provider via cellular networks. Therefore, some of the non-seed users can be satisfied by opportunistic delivery and their cellular bandwidth can be offloaded. The efficiency of opportunistic offloading can be guaranteed because content subscribers can usually tolerate different delays between the content generation time and the user access time~\cite{YL-2011, GT-2010}, and meanwhile, they can frequently meet each other due to homophily and locality~\cite{WPDAZ-2010, RBCGA-2011}.

Most of the existing approaches of opportunistic offloading focus on maximizing the overall system performance, such as the total bandwidth offloaded from the cellular networks, or the average delay of subscribers~\cite{HHKMPS-2010, ICM-2009, LQJHWC-2014}. In order to maximize the system performance, these approaches require the seed users to download the content right after the content is posted online, and require all mobile users to unconditionally share the content with each other when they are in proximity. However, since both cellular and D2D transmissions may involve monetary cost of mobile users, these assumptions may not be hold for selfish users, who only consider their individual payoffs rather than the overall system performance. First, there is no incentive for a selfish user to join such approaches as a seed user, since seed users have to download each content update via cellular networks, while at the same time, to deliver the content to non-seed users via opportunistic communications. Second, for the opportunistic sharing process, selfish users may have preference on the users that access the content at different times. Usually, they tend to communicate with the users with early access times, since these users have a high probability to deliver the content, and thus, save the expensive cellular transmission. At last, privacy issues may also restrict opportunistic communications. Usually, the users can help each other by opportunistic sharing only if they have a social relationship in the considered SNS. In order to capture the dynamic behaviors of practical users and analyze its influence on the overall system performance, it is necessary to reconsider opportunistic offloading from a new perspective that treats each mobile user as an independent decision maker.

The idea of self-deciding users naturally leads us to game theory, in which all users are treated as rational players that focus on maximizing their own utilities. In this paper, we consider opportunistic offloading in D2D-enhanced cellular networks as a \emph{network formation game}. Note that we specify on D2D-enhanced cellular networks only because they can provide a system-level short-range communication solution, which avoids the tedious discussions on the establishment of opportunistic connections as in other cases. The contributions of this paper are as follows:
\begin{enumerate}
\item First, we propose a user-centered opportunistic offloading approach, in which each user can download the content at arbitrary times via cellular networks, and meanwhile, independently decide whether to opportunistically share the content with another user. Specifically, we formulate the problem as a \emph{network formation game}, in which the users autonomously form a cooperative network, and promise D2D sharing with their adjacent users.
\item Second, we provide the stability analysis of the dynamics of network formation and specify the conditions to form a \emph{pairwise stable} network. Based on the network formation analysis, we propose a low-complex distributed network formation algorithm for the mobile users to decide their D2D transmissions in practice.
\item Third, simulation results show that the selfishness of mobile users can highly degrade the offloading efficiency, and gap is highly affected by the cost ratio between the cellular transmission and the D2D transmission, as well as the users' access delays and mobility patterns.
\end{enumerate}

The rest of the paper is organized as follows. In Section~\uppercase\expandafter{\romannumeral2}, we provide the related works in the literature. In Section~\uppercase\expandafter{\romannumeral3}, we present the system model by characterizing the users in terms of access delay, mobility pattern, and individual payoff. In Section~\uppercase\expandafter{\romannumeral4}, we formulate the considered problem as a \emph{network formation game} so as to analyze the users' selfish behaviors. In Section~\uppercase\expandafter{\romannumeral5}, we specify the conditions under which the users will form a \emph{pairwise stable} network, and propose a practical network formation algorithm for the users can make selfish decisions based on their historical records. In Section~\uppercase\expandafter{\romannumeral6}, we provide the simulation results of algorithm convergence, offloading efficiency, and user payoff. In Section~\uppercase\expandafter{\romannumeral7}, we discuss some extensions of the proposed game-theoretical approach, and conclude the paper.

%%%%%%%%%%%%%%%%%%%%%%%%
\section{Related Works}%
%%%%%%%%%%%%%%%%%%%%%%%%

\subsection{Opportunistic Offloading}

In~\cite{HHKMPS-2010, ICM-2009, LQJHWC-2014}, the proposed approaches assume a central controller to decide the seed users, and all users are required to unconditionally share its copy with every other user in proximity. As we noted, these approaches focus on the system performance and do no consider the selfishness of practical users, which is completely from the game-theoretic perspective of this paper. In \cite{ZGCH-2014}, the authors propose an incentive framework to motivate the subscribers to leverage their delay tolerance for traffic offloading, in which the subscribers act as sellers to sell their tolerant delay by submitting a price to the content provider, and the content provider acts as the buyer to buy the delay tolerance from the subscribers. This mechanism gives mobile users the freedom to decide their tolerant delays, but still requires unconditional seeding and sharing as in previous approaches~\cite{ICM-2009, HHKMPS-2010, LQJHWC-2014}.

In our previous paper~\cite{SWSH-2014}, we consider to use D2D communications to facilitate real-time multicasting services, such as football games broadcasting, by constructing a dynamic ad hoc network to detour the traffic from the base station. In this paper, we consider a different delay-tolerant publish/subscribe service, in which D2D links are utilized as opportunistic connections rather than instant ad hoc links. Therefore, this paper is not a simple extension of~\cite{SWSH-2014}.

\subsection{Data Forwarding and Data Dissemination in DTNs}

Data forwarding and data dissemination are highly related to the topic of opportunistic offloading. They all share the idea of delay tolerant networks and use short-range communications to improve the system performance. However, in data forwarding, there exists no cellular infrastructure and the problem is how to efficiently deliver the data from the source to the destination via opportunistic forwarding~\cite{LDS-2003, JLW-2007, YHCC-2007, HCY-2011}, rather than offloading traffic from cellular networks. In data dissemination, or data spreading, the problem is how to predict and accelerate the propagation of information in social networks, considering how we initially push the information, as well as how social impacts motivate the agents to re-share it~\cite{WD-2007,RBCGA-2011}. Recently, due to the rise of mobile devices, some studies combine the online SNS and the offline mobile social networks so as to facilitate the information spreading process, while at the same time, reduce the cellular bandwidth by opportunistic communications between mobile users~\cite{ZSSDH-2013, WCHWK-2014}. Compared to opportunistic offloading, data dissemination does not have a strict publish/subscribe model, that is, there is not a specific set of subscribers that require the content. The ``subscribers" in data dissemination are gradually decided as the content is re-shared among agents in the online SNS. Although we do not consider data dissemination in this paper, we will show that the proposed game theory model and network formation algorithm can be extended to such scenarios.

%%%%%%%%%%%%%%%%%%%%%%%%%%%%%%%%%%%%%%%%%%%%%%%
\section{System Model}%
%%%%%%%%%%%%%%%%%%%%%%%%%%%%%%%%%%%%%%%%%%%%%%%

\begin{figure}[!t]
\centering
\includegraphics[width=3.2in]{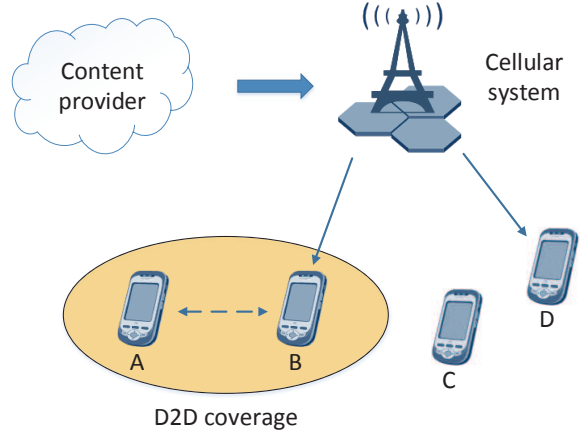}
\caption{Illustration of opportunistic offloading in a D2D-enhanced network.} \label{system}
\end{figure}

We consider a delay-tolerant publish/subscribe service in a D2D-enhanced cellular network, in which a content provider pushes information to $N$ subscribers, the set of which is given by $\Omega = \{1,2,\ldots,N\}$, and the subscribers access the content with different delays. Each subscriber is an independent user that can make long-term D2D sharing agreements with other users, that is, if an agreement is made between user $i$ and user $j$, they will opportunistically share the content with each other using D2D communications when they are in proximity. If some user fails to achieve the content via opportunistic sharing before its tolerable delay, it will instantly download the content via cellular networks. As illustrated in Fig.~\ref{system}, users B and D download the content via cellular networks, user A opportunistically receives the content from user B who has an agreement with user A, and user C still waits for opportunistic sharing from other users. We assume the cellular network is in charge of the functions of user discovery and D2D establishment, and the associated signaling is negligible to size of the subscribed content. In the following part of this section, we will characterize the users' access delays and mobility patterns, and then formulate their individual gain and cost in opportunistic offloading.

\subsection{Access Delay}

In publish/subscribe applications, subscribers may access the service at different frequencies, and thus, endure different delays~\cite{YL-2011, GT-2010}. The delay tolerance of a subscriber is usually determined by the user's lifestyle, and it can be quantified by the \emph{access delay}, which is defined as the time between the content generation time and the user's access time. We assume the content is always generated at $t=0$, and denote the access delay of user $i$ as $a_i$. In general, $a_i$ is a random variable defined in interval $[0,+\infty)$, with its probability distribution function (PDF) determined by user $i$'s lifestyle. In order to model the various PDFs of $a_i$, we use the Weibull distribution to profile the access delays~\cite{GT-2010}:
\begin{equation} \label{access-delay}
Pr\{a_i = t\} = f^A_i(t, \lambda_i, k_i) = \frac{k_i}{\lambda_i} \left(\frac{t}{\lambda_i}\right)^{k_i-1}\exp\left\{-\left(\frac{t}{\lambda_i}\right)^{k_i}\right\}, t \ge 0,
\end{equation}
where the parameter $k_i>0$ decides the shape of $a_i$'s PDF, i.e., $a_i$ concentrates on higher values if $k_i$ is larger, and the parameter $\lambda_i>0$ is a scale parameter that roughly decides the expected value of $a_i$, i.e., $a_i$ is larger when $\lambda_i>0$ is larger. We assume that the access times $\{a_i  | i \in \Omega\}$ are stochastically independent.

\subsection{Contact Graph}

It has been studied that mobile users have different mobility patterns~\cite{ZNKT-2007, KLV-2010}. For a pair of users, the time line can be divided into contact times and inter-contact times. Contact times are when the users stay in D2D range, and could therefore perform opportunistic sharing. Inter-contact times are times between two consecutive contacts. Referring to~\cite{ZNKT-2007, KLV-2010}, for a characteristic time in the order of half a day, we assume that inter-contact time of any two users $i$ and $j$, denoted by $b_{i,j}$, follow the Pareto distribution:
\begin{equation} \label{inter-contact}
Pr\{b_{i,j} = t\} = f^B_{i,j}(t, \tau_{i,j}, \alpha_{i,j}) = \frac{\alpha_{i,j} t^{\alpha_{i,j}}}{t^{\alpha_{i,j}+1}}, t \ge \tau_{i,j},
\end{equation}
where the parameter $\alpha_{i,j}>1$ decides the shape of $b_{i,j}$'s PDF, i.e., $b_{i,j}$ concentrates more on low values if $\alpha_{i,j}$ is larger, and the scale parameter $\tau_{i,j}>0$ identifies the minimum possible value of $b_{i,j}$. The complementary cumulative distribution function of $b_{i,j}$ is given by:
\begin{equation}
Pr\{b_{i,j} > t\} = \bar{F}^B_{i,j}(t, \tau_{i,j}, \alpha_{i,j}) = \left(\frac{\tau_{i,j}}{t}\right)^{\alpha_{i,j}}, t \ge \tau_{i,j}.
\end{equation}
Then, the probability that user $i$ and $j$ meet each other at least once during time period $\Delta t$ is given by:
\begin{equation} \label{meet}
P^{\text{meet}}_{i,j}(\Delta t) = 1 - \left(\frac{\tau_{i,j}}{\Delta t}\right)^{\alpha_{i,j}}.
\end{equation}

We assume that the inter-contact times $\{b_{i,j} | i,j \in \Omega\}$ are stochastically independent. Also, we assume that the content can be instantly transmitted via D2D communications, and the ``effective" contact time can be negligible compared to the inter-contact time. Thus, the reciprocal of the expected inter-contact time, $\frac{\alpha_{i,j} - 1}{\alpha_{i,j} \tau_{i,j}}$, simply indicates the associated contact rate between user $i$ and user $j$. In the literature, the mobility characteristic is usually represented by a complete, weighted and undirected graph of all involved users, referred to as \emph{contact graph}~\cite{HHKMPS-2010, ICM-2009, YHCC-2007, HCY-2011}, where the weight of the edge between user $i$ and user $j$ is the associated contact rate $\frac{\alpha_{i,j} - 1}{\alpha_{i,j} \tau_{i,j}}$.

\subsection{Gain and Cost}

To analyze opportunistic offloading from the user's point of view, we need to specify the gain and cost of a mobile user in this mechanism. Note that we assume a user $i$ will enjoy the content only after its access delay $a_i$, no matter whether it receives the content via cellular networks or via D2D communications. Thus, there is no difference in user experience between the cellular service or the opportunistic offloading service. The only difference is that opportunistic offloading may delivery the content through D2D communications, and thus, save the expensive cost of cellular downloading for user $i$. We assume the monetary cost of downloading the content via cellular networks is homogeneous for all subscribers, and denote it by $v_{c}$. Therefore, the gain of user $i$ in opportunistic offloading is given by:
\begin{equation} \label{gain}
g_i = v_{c} \cdot n^{to}_i,
\end{equation}
where $n^{to}_i$ is the expected number of D2D delivery from other users to user $i$.

In opportunistic offloading, user $i$ may transmit to or receive from other users via D2D communications, which may cause monetary and power cost. Here, we simply assume that the cost of D2D transmission is homogenously $v_{d}$, and that the cost of D2D reception is homogenously zero. Usually, we have $v_{d} \ll v_c$. Therefore, the cost of user $i$ in opportunistic offloading is given by:
\begin{equation} \label{cost}
c_i = v_{d} \cdot n^{from}_i,
\end{equation}
where $n^{from}_i$ is the expected number of D2D delivery from user $i$ to other users.

%%%%%%%%%%%%%%%%%%%%%%%%%%%%%%%%%%%%%%%%%%%%%%%%%%%%%%%%%%%%%%%
\section{Social Data Offloading as a Network Formation Game}%
%%%%%%%%%%%%%%%%%%%%%%%%%%%%%%%%%%%%%%%%%%%%%%%%%%%%%%%%%%%%%%%

In this section, we introduce the proposed game-theoretic model, \emph{network formation games}, using which the individual profits of mobile users rather than the overall system performance are considered to have the first priority in opportunistic offloading.

\subsection{Network Formation Games}

\emph{Network formation games} describe how a set of players can cooperate with each other by forming a network and thereby increase their utilities~\cite{PR-2013, SHDHB-2009}. As for the considered opportunistic offloading problem, the selfish mobile users are represented by the players in the game, and their agreements on opportunistic sharing are represented by the links between players. For each user, its gain and cost can be captured by a function that depends on the entire network structure formed by all the users. Formally, we define a network formation game as follows:

\begin{definition}\label{NFG-definition}
A network formation game is uniquely defined by a set of players $\Omega$ and their network payoff functions $\Phi_i(\cdot), i \in \Omega$. For each player $i \in \Omega$ and each network $G$ formed by all the players (i.e., $V(G)=\Omega$), $\Phi_i(G)$ is the payoff of player $i$ in network $G$.
\end{definition}

Given that the opportunistic sharing agreement is a bilateral contact that indicates a mutual obligation, the links between players are undirected edges. Also, for obvious reasons, loops and multiple edges are not allowed in the network. Thus, the network $G$ formed by the players is an undirected simple graph, and we denote by $\mathbb{G}$ as the set of all such graphs. For any network $G \in \mathbb{G}$, we define $G^i$ as the maximal connected subgraph of $G$ that contains user $i$, i.e. $G^i \subseteq G$ is a connected subgraph with $i \in V(G^i)$, and for any connected subgraph $G' \subseteq G$ with $i \in V(G')$, we have $V(G') \subseteq V(G^i)$.

\subsection{Network Payoff Functions}

To complete the proposed network formation game, we need to specify the network payoff function $\Phi_i(G)$ for each player $i \in \Omega$. In general, the payoff function of player $i \in \Omega$ can be any function that increases with the gain $g_i$ in (\ref{gain}) and decreases with the cost $c_i$ in (\ref{cost}). Here, for any network $G \in \mathbb{G}$, we simply define the payoff of player $i$ as follows:
\begin{equation} \label{formal-payoff}
\Phi_i(G) = g_i(G) - c_i(G).
\end{equation}
Note that only the users that are connected to user $i$ can effect user $i$'s delivery performance. The above payoff function can also be written as $\Phi_i(G_i) = g_i(G^i) - c_i(G^i)$.

However, even we restrict our consideration to $G^i$, calculating of $\Phi_i(G^i)$ is extremely difficult due to the statistical correlation of the players in $V(G^i)$. In order to simply the calculation, we estimate the value of $\Phi_i(G^i)$ by considering only one-hop D2D transmissions, and denote the estimated value as $\bar{\Phi}_i(G^i)$. For any network $G \in \mathbb{G}$ and any player $i \in \Omega$, we denote by $S^i$ as the set of player $i$'s adjacent players, i.e., $S^i = \{j \in \Omega | e_{i,j} \in E(G)\}$. We remark the players in $S^i$ as $\alpha_1, \alpha_2, \ldots, \alpha_{|S_i|}$, and without loss of generality, we assume that the access delays of these players satisfy $a_{\alpha_1} \le \ldots a_{\alpha_{h}} \le a_i \le a_{\alpha_{h+1}} \le \ldots \le a_{\alpha_{|S^i|}}$.

First, we note that only the players $\{\alpha_l\}, 1 \le l \le h$ can download the content before time $a_i$ and thus can share the content with player $i$ via one-hop transmissions. For any player $\alpha_l, 1 \le l \le h$, it can share the content to player $i$ between the time interval $[a_{\alpha_l},a_i)$, if there is no other player that receives the content and meets player $i$ at an earlier time. Thus, the probability that player $i$ receives the content from player $\alpha_l$ is given by
\begin{equation} \label{P-alphal-i}
P_{\alpha_l,i} = \iiint\nolimits_{0}^{+\infty} \rho_{\alpha_l,i} \cdot f^A_{i}(a_{i}) f^A_{\alpha_1}(a_{\alpha_1}) \ldots f^A_{\alpha_{|S^i|}} (a_{\alpha_{|S^i|}}) \cdot da_i da_{\alpha_1} \ldots da_{\alpha_{|S^i|}},
\end{equation}
with
\begin{equation} \label{rho-alphal-i}
\rho_{\alpha_l,i} = \int\limits_{a_{\alpha_l}}^{a_i} \left\{\prod\limits_{l'=1,l' \ne l}^{h} \left[ 1 - P^{\text{meet}}_{i,\alpha_{l'}}(t-a_{\alpha_{l'}}) \right]\right\} f^B_{i,\alpha_l}(t-a_{\alpha_{l}}) dt,
\end{equation}
where $f^A_{i}(\cdot,\cdot,\cdot)$, $f^B_{i,j}(\cdot, \cdot)$ and $P^{\text{meet}}_{i,j}(\cdot)$ are given as in (\ref{access-delay}), (\ref{inter-contact}) and (\ref{meet}), respectively, and some parameters are omitted for the sake of conciseness.

Also, player $i$ may download the content directly from the base station at time $a_i$, if no other players can deliver the content before $a_i$, the probability of which is given by:
\begin{equation} \label{P-i-i}
P_{i,i} = \iiint\nolimits_{0}^{+\infty} \rho_{i,i} \cdot f^A_{i}(a_{i}) f^A_{\alpha_1}(a_{\alpha_1}) \ldots f^A_{\alpha_{|S^i|}} (a_{\alpha_{|S^i|}}) \cdot da_i da_{\alpha_1} \ldots da_{\alpha_{|S^i|}},
\end{equation}
with
\begin{equation} \label{rho-i-i}
\rho_{i,i} = \prod\limits_{l=1}^{h} \left[ 1 - P^{\text{meet}}_{i,\alpha_{l}}(a_i-a_{\alpha_{l}}) \right],
\end{equation}
where $f^A_{i}(\cdot,\cdot,\cdot)$ and $P^{\text{meet}}_{i,j}(\cdot)$ are given as in (\ref{access-delay}) and (\ref{meet}), respectively, with some parameters omitted. Note that the one-hop assumption narrows our consideration down to the adjacent player in $S_i$, and more importantly, eliminates the statistical correlation between them, such that we can calculate the effect of different players separately as in (\ref{P-alphal-i}) and (\ref{P-i-i}). Naturally, we have $\sum\nolimits_{l=1}^{h} P_{\alpha_l,i} + P_{i,i} = 1$.

Given (\ref{P-alphal-i}) and (\ref{P-i-i}), we estimate the gain, cost, and thus the payoff of player $i$ as follows:
\begin{equation} \label{gain-es}
\bar{g}_i(S^i) = v_c \cdot (1 - P_{i,i}),
\end{equation}
\begin{equation} \label{cost-es}
\bar{c}_i(S^i) = v_d \cdot \sum\limits_{j \in S^i} P_{i,j},
\end{equation}
\begin{equation} \label{utility-es}
\bar{\Phi}_i(S^i) = \bar{g}_i(S^i) - \bar{c}_i(S^i).
\end{equation}

In each realization of opportunistic offloading, the content data is delivered through a tree structure rooted at the cellular system, which guarantees the total number of transmissions, including both cellular and D2D transmissions, is fixed at $|\Omega|$. Thus, the multi-hop D2D transmissions, which are not considered in our calculation, will increase the number of D2D transmissions, and at the same time, decrease the same amount of cellular transmissions. Due to the fact that $v_d \ll v_c$, we have that the overall payoff of the players in $G^i$ is underestimated, i.e., $\sum\nolimits_{i \in V(G^i)} \bar{\Phi}_i(S^i) \le \sum\nolimits_{i \in V(G^i)} \Phi_i(G^i)$. Therefore, (\ref{utility-es}) can be seen as the payoff function in the worst case scenario, and the network is expected to have a better performance.

%%%%%%%%%%%%%%%%%%%%%%%%%%%%%%%%%%%%%%%%
\section{Dynamics of Network Formation}%
%%%%%%%%%%%%%%%%%%%%%%%%%%%%%%%%%%%%%%%%

In this section, we consider how the players can automatically form a network. First, we specify the rules of network formation, which restricts what a player can be do in the network formation process. Then, we analyze the stability of the outcome network under such network formation games. At last, we consider the practical issues and propose a low-complex algorithm that can be distributively performed by mobile users.

\subsection{Rules of Network Formation}

The players in network formation games are independent decision makers who seek to maximize their individual payoffs by choosing selfish strategies. Note that it does not imply that the players cannot cooperate with each other. In many cases, the players are assumed to form cooperative groups, or coalitions, to increase the payoffs of coalition members, and such network formation games are often referred to as coalitional graph games~\cite{SHDHB-2009}. In such games, the players in a coalition are treated as an entity and their payoffs are jointly decided by the entire coalition. In fact, for any network $G \in \mathbb{G}$ of our proposed game, any player $i$ can be seen as a member of coalition $V(G^i)$, and player $i$'s payoff $\Phi(G^i)$ is decided by the cooperation among all coalition members, i.e., the structure of $G^i$. However, this coalitional approach can be problematic in opportunistic offloading, as most of the subscribers are completely strangers without any communications, and the size of a potential coalition could be extremely large due to the huge amount of users in a publish/subscribe service. Thus, we assume that the players individually choose their own strategies in the network formation process.

For any network $G \in \mathbb{G}$, the strategy of any player $i \in \Omega$ can be given by the set of player $i$'s adjacent links $\{e_{i,j} | e_{i,j} \in E(G)\}$, which indicates all the opportunistic sharing agreements that player $i$ makes with other players. During the network formation process, each player can dynamically add links to the network and subtract the existing links from the network. Note that an opportunistic sharing agreement indicates a mutual obligation of both ends, a link can exist only when it benefits both ends, and a link should be removed if any end cannot benefit from it. Formally, we give the rules of network formation as follows~\cite{JW-1996}:
\begin{enumerate}
\item add a link $e_{i,j}$ to $G$ requires that both players $i$ and $j$ agree to add the link, i.e., $\Phi_i(G \cup e_{i,j}) > \Phi_i(G)$ \emph{and} $\Phi_j(G \cup e_{i,j}) > \Phi_j(G)$, and we denote by $E^+(G)$ as the set of all links that can be added to $G$;
\item subtract a link $e_{i,j}$ from $G$ requires that player $i$ or player $j$ or both agree to subtract the link, i.e., $\Phi_i(G \backslash e_{i,j}) > \Phi_i(G)$ or $\Phi_j(G \backslash e_{i,j}) > \Phi_j(G)$, and we denote by $E^-(G)$ as the set of all links that can be subtracted from $G$;
\item link addition or link subtraction takes place one link at a time.
\end{enumerate}
The third rule implies that the players take random turns to alter their strategies, and in each turn, the corresponding player can only add a new link, or subtract an existing link. With this rule, the size of the player's action space is restricted to $N$. If we remove this rule, the player can add and subtract multiple links, and the size of its action space increases to $2^N$. When $N$ is large, it will bring combinatorial difficulty for the player to decide a feasible action. Therefore, we adopt the third rule in our proposed network formation game.

\subsection{Conditions for Pairwise Stable Networks}

In network formation games, the players keep changing their strategies to maximize their individual payoffs. The dynamics may end up to a network in which all players arrive their optimal strategies at the same time, i.e., no player has the incentive to change its current strategy and thereby improve its payoff, assuming all the other players stick to their current strategies. Such networks are called \emph{Nash-stable} networks, or Nash networks~\cite{BJ-2006}. In our proposed game, the player's actions are restricted by the network formation rules, and Nash stability is replaced by the following \emph{pairwise stability}~\cite{JW-1996}:

\begin{definition}\label{pairwise-stability}
A network $G$ is pairwise stable if both of the following conditions are satisfied:
\begin{enumerate}
\item for any pair of players $i,j \in \Omega$ with edge $e_{i,j} \notin E(G)$, we have $\Phi_i(G \cup e_{i,j}) \le \Phi_i(G)$ or $\Phi_j(G \cup e_{i,j}) \le \Phi_j(G)$, or in other words, $E^+(G) = \emptyset$;
\item for any pair of players $i,j \in \Omega$ with edge $e_{i,j} \in E(G)$, we have $\Phi_i(G \backslash e_{i,j}) < \Phi_i(G)$ and $\Phi_j(G \backslash e_{i,j}) < \Phi_j(G)$, or in other words, $E^-(G) = \emptyset$.
\end{enumerate}
\end{definition}

Thus, a network is pairwise stable, if there is no incentive for any pair of players to add a link to the existing network and there is no incentive for any player who is party to a link in the existing network to dissolve or remove the link.

\begin{figure}[!t]
\centering
\includegraphics[width=4.8in]{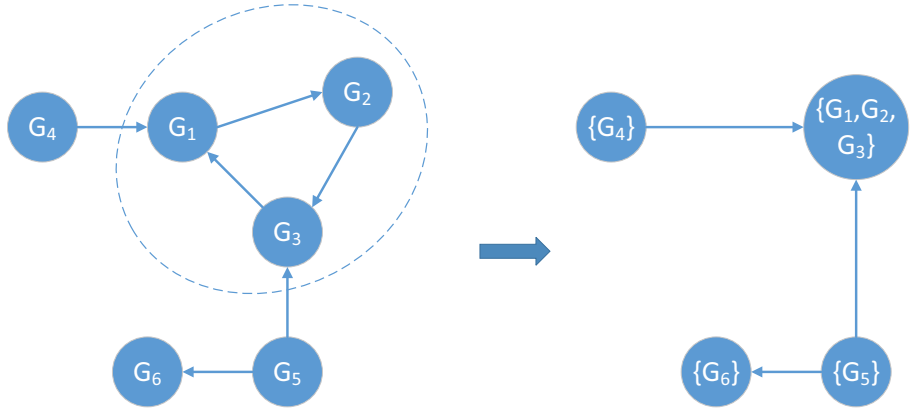}
\caption{The rule supernetwork $\Upsilon$ on the left side has a circuit consists of nodes $G_1,G_2$ and $G_3$, i.e., $G_1 \equiv_p G_2 \equiv_p G_3$. The path equivalence partition $\mathcal{P}$ is given by $\{ \{G_1,G_2,G_3\}, \{G_4\}, \{G_5\}, \{G_6\} \}$. The acyclic supernetwork $\Gamma$ on the right side has two basins $\{G_1, G_2, G_3\}$ and $\{G_6\}$, and $G_6$ is the only pairwise stable network.} \label{supernetwork}
\end{figure}

Given the definition of pairwise stability and the proposed rules for network formation, it is easy to see that, if the network formation process converges to a final network $G \in \mathbb{G}$, then the network $G$ must be pairwise stable. However, the existence of pairwise stable networks and the convergency of network formation rules are still unknown. In this subsection, we will specify the conditions under which the pairwise stable network exists and the players are guaranteed to converge to such networks by using the proposed network formation rules.

By viewing each network $G \in \mathbb{G}$ as a node in a larger network, we define a \emph{rule supernetwork} of the proposed network formation game as follows~\cite{PW-2007}:

\begin{definition}\label{rule-suppernetwork}
A rule supernetwork is a directed simple graph $\Upsilon$ where $V(\Upsilon)=\mathbb{G}$ and $e_{G_i,G_j} \in E(\Upsilon), G_i,G_j \in \mathbb{G}$ if either of the following conditions is satisfied:
\begin{enumerate}
\item $|E(G_i)-E(G_j)|=1$ and $E(G_i)-E(G_j) \subseteq E^-(G_i)$;
\item $|E(G_j)-E(G_i)|=1$ and $E(G_j)-E(G_i) \subseteq E^+(G_i)$.
\end{enumerate}
\end{definition}

Thus, each edge $e_{G_i,G_j} \in E(\Upsilon)$ indicates that network $G_i$ can be transformed to network $G_j$ with the proposed network formation rules, i.e., by adding a link in $E^+(G_i)$ or subtracting a link in $E^-(G_i)$. The entire rule supernetwork $\Upsilon$ specifies how the network formation rules transform a network to another network.

Given the rule supernetwork $\Upsilon$, we now define the \emph{path dominance} relation $\succeq_p$ and the \emph{path equivalence} relation $\equiv_p$ on networks $\mathbb{G}$ as follows~\cite{PW-2007}:

\begin{definition}\label{path}
Given the rule supernetwork $\Upsilon$, network $G' \in \mathbb{G}$ path dominates network $G \in \mathbb{G}$, written as $G' \succeq_p G$, if there exists a finite sequence of networks $\{G_k\}_{k=0}^{h}$ in $\mathbb{G}$ with $G_0=G$ and $G_h=G'$ such that edge $e_{G_{k-1},G_k} \in E(\Upsilon)$ for $k=1,2,\ldots,h$. And network $G' \in \mathbb{G}$ is path equivalent to network $G \in \mathbb{G}$, written as $G' \equiv_p G$, if $G' \succeq_p G$ and $G \succeq_p G'$.
\end{definition}

Thus, $G' \succeq_p G$ indicates that there exists a path in $\Upsilon$ that starts from node $G$ and ends at node $G'$, and $G' \equiv_p G$ indicates that there exists a circuit in $\Upsilon$ that contains both nodes $G$ and $G'$. Therefore, $G' \succeq_p G$ implies that network $G$ can automatically transform to network $G'$, and $G' \equiv_p G$ implies that networks $G$ and $G'$ can transform between each other.

Note that the path equivalence relation is reflexive, symmetric and transitive. We can classify the networks $\mathbb{G}$ by $\equiv_p$ and achieve a partition $\mathcal{P}$ of $\mathbb{G}$, i.e., $\bigcup\nolimits_{\mathbb{C} \in \mathcal{P}} \mathbb{C} = \mathbb{G}$ and $\forall \mathbb{C}, \mathbb{C}' \in \mathcal{P}, \mathbb{C} \cap \mathbb{C}'=\emptyset$. Given the initial $\mathcal{P} = \emptyset$, we can specify this partition by sequentially adding the networks $\mathbb{G}$ as follows: for network $G \in \mathbb{G}$, if $\exists \mathbb{C} \in \mathcal{P}$ such that $\exists G' \in \mathbb{C}$ and $G' \equiv_p G$, then set $\mathbb{C} = \mathbb{C} \cup G$, if not, then set $\mathcal{P} = \mathcal{P} \cup \{G\}$. Note that each element $\mathbb{C} \in \mathcal{P}$ represents a different circuit in $\Upsilon$, and due to the above process, no networks belong to different elements of $\mathcal{P}$ are path equivalent, i.e., for any networks $G \in \mathbb{C}, \mathbb{C} \in \mathcal{P}$ and $G' \in \mathbb{C}', \mathbb{C}' \in \mathcal{P}$, if $G \equiv_p G'$, then $\mathbb{C}=\mathbb{C}'$. We refer to $\mathcal{P}$ as the \emph{path equivalence partition}.

By viewing each element in $\mathcal{P}$ as a node, we define an \emph{acyclic supernetwork} of the proposed network formation game as follows:

\begin{definition}\label{scyclic-suppernetwork}
Given the rule supernetwork $\Upsilon$ and the path equivalence partition $\mathcal{P}$, an \emph{acyclic supernetwork} $\Gamma$ is a simple directed network where $V(\Gamma) = \mathcal{P}$, and edge $e_{\mathbb{C},\mathbb{C}'} \in E(\Gamma)$ if and only if there exists networks $G \in \mathbb{C}$ and $G' \in \mathbb{C}'$ such that $G' \succeq_p G$.
\end{definition}

Thus, the acyclic supernetwork $\Gamma$ indicates how the network transforms from one path equivalent network set to another path equivalent network set in $\mathcal{P}$. We say $\Gamma$ is acyclic because, if not, e.g., there exists a circuit containing $\mathbb{C}, \mathbb{C}' \in V(\Gamma)$ in $\Gamma$, then for any networks $G \in \mathbb{C}$ and $G' \in \mathbb{C}'$, the path from $\mathbb{C}$ to $\mathbb{C}'$ and the path from $\mathbb{C}'$ to $\mathbb{C}$ indicate $G' \succeq_p G$ and $G \succeq_p G'$, respectively, contradicting the fact that $G'$ and $G$ are not path equivalent. Therefore, once a network transforms to another network that belongs to a different node in $\Gamma$, there is no way for the players to return to the original network. We define the nodes in $\Gamma$ that have no out-going edges as \emph{basins}, i.e., $\mathbb{C} \subseteq V(\Gamma)$ is a basin if $e_{\mathbb{C},\mathbb{C}'} \notin E(\Gamma)$ for any node $\mathbb{C}' \in V(\Gamma)$. Thus, a basin is a set of networks to which the network formation process might tend and from which there is no escape. A network in a basin can only transform to other networks in the same basin. Particularly, if the cardinality of a basin $\mathbb{B}$ is $1$, i.e., $|\mathbb{B}|=1$, then the only network in this basin cannot transform to any other network, and thus, it is pairwise stable according to Definition~\ref{pairwise-stability}. Fig.~\ref{supernetwork} shows an example to illustrate the above concepts.

\begin{theorem}\label{stability}
Given a network formation game with basins $\mathcal{B}=\{\mathbb{B}_1, \mathbb{B}_2, \ldots \mathbb{B}_m\}$, where basin $\mathbb{B}_k$ contains $|\mathbb{B}_k|$ many networks, the following statements are true:
\begin{enumerate}
\item There always exists basins in a network formation game, i.e., $\mathcal{B} \ne \emptyset$, and after $L < \infty$ finite number of adding and subtracting a link, the following networks $G_L, G_{L+1},\ldots$ will be contained by a basin of the game, i.e., $\exists \mathbb{B} \in \mathcal{B}$ such that $G_k \in \mathbb{B}$ for $k \ge L$.
\item Pairwise stable networks exist if and only if some basins are singletons, i.e., $\exists 1 \le k \le m, |\mathbb{B}_k|=1$, and the set of all pairwise stable networks $\mathbb{PS}$ is given by the networks in such basins, i.e., $\mathbb{PS} = \bigcup\nolimits_{|\mathbb{B}_k|=1} \mathbb{B}_k$.
\item The network converges to a pairwise stable network after finite number of network formation operations, if and only if all the basins are singletons, i.e., $|\mathbb{B}_k|=1, k=1,2,\ldots m$.
\end{enumerate}
\end{theorem}

\begin{proof}

1) Supposing $\mathcal{B} = \emptyset$, then according to the definition of basins, each node in the acyclic supernetwork $\Gamma$ has an out-going edge. We remark the nodes in $\Gamma$ and denote by $e_{\mathbb{C}_k, \mathbb{C}_{k+1}}$ as the out-going edge of node $\mathbb{C}_k, k=1,2,\ldots$. Then, we will have a path of $\Gamma$ as follows:
\begin{equation}
\mathbb{C}_1 \stackrel{e_{\mathbb{C}_1, \mathbb{C}_2}}{\longrightarrow} \mathbb{C}_2, \stackrel{e_{\mathbb{C}_2, \mathbb{C}_3}}{\longrightarrow} \mathbb{C}_3 \stackrel{e_{\mathbb{C}_3, \mathbb{C}_4}}{\longrightarrow} \mathbb{C}_4 \stackrel{e_{\mathbb{C}_4, \mathbb{C}_5}}{\longrightarrow} \mathbb{C}_5 \ldots
\end{equation}
Note that $\Gamma$ is a graph with finite number of nodes. There exists $k>1$ such that $\mathbb{C}_1 = \mathbb{C}_k$. Thus, the above path is a circuit containing nodes $\mathbb{C}_1, \mathbb{C}_2, \ldots \mathbb{C}_k$, which contradicts the fact that $\Gamma$ is acyclic. Therefore, we have $\mathcal{B} \ne \emptyset$.

Consider a network $G$ that belongs to a non-basin node $\mathbb{C}$ in the acyclic supernetwork $\Gamma$, i.e., $\mathbb{C} \in V(\Gamma), \mathbb{C} \notin \mathcal{B}$. Then, there exists an out-going edge $e_{\mathbb{C}, \mathbb{C}'} \in E(\Gamma)$ from $\mathbb{C}$ to $\mathbb{C}'$, and thus, a path from network $G$ to another network in $\mathbb{C}'$ in $\Upsilon$. Due to the acyclic property of $\Gamma$, there does not exist a path from $\mathbb{C}'$ to $\mathbb{C}$ in $\Gamma$, and thus, neither a path from any network in $\mathbb{C}'$ to network $G$ in $\Upsilon$. Thus, when network $G$ randomly performs network formation operations, it will finally transform to a network of an adjacent node, e.g., $G' \in \mathbb{C}'$, and it cannot transform back to $G$ from $G'$. Therefore, the network will jump from node to node non-repetitively, until it drops into a basin node ($\mathcal{B} \ne \emptyset$), where there is no way out, and spins inside the basin from one network to another.

2) According to Definition~\ref{pairwise-stability}, a network is pairwise stable if it cannot transform to any another network according to the network formation rules. First, we note that all networks of the non-basin nodes are not pairwise stable, since they can transform to networks of their adjacent nodes as mentioned above. Second, we note that the networks of a basin node can transform between each other, since they are path equivalent. Thus, only the network of a basin with cardinality $1$ cannot transform to any other network, and thus is pairwise stable. Therefore, we have the set of all pairwise stable networks given by $\mathbb{PS} = \bigcup\nolimits_{|\mathbb{B}_k|=1} \mathbb{B}_k$.

3) As we proved in 1), after a finite number of network formation operations, the network will finally drop into a basin $\mathbb{B} \in \mathcal{B}$. If the basin contains multiple elements, then the network will not converge but spins among the networks in this basin. If all basins have cardinality $1$, then the network will converge no matter which basin it drops into. Also, as we proved in 2), the network of the basin with cardinality $1$ is pairwise stable. Therefore, the network will converge to a pairwise stable network if and only if $|\mathbb{B}_k|=1, k=1,2,\ldots m$.
\end{proof}

In summary, we prove that a network formation game will generally converge to particular sets of networks (basins), but may not converge to a particular network. The convergency is guaranteed only if all such particular sets have cardinality $1$ ($|\mathbb{B}_k|=1, k=1,2,\ldots m$), and once the conditions are satisfied, we can guarantee that it converges to a pairwise stable network.

\subsection{Data-Based Network Formation Algorithm}

\begin{algorithm}
\caption{Data-Based Network Formation Algorithm}
\label{NFA}
\begin{algorithmic} [1]
\REQUIRE~\\
The vector $\boldsymbol{d}^{i}$ that records the number of D2D transmissions from user $i$ in this period. \\
The vector $\boldsymbol{b}^{i}$ that records the number of D2D transmissions to user $i$ in this period. \\
The existing agreements between user $i$ and other users $E^i = \{e_{i,j} \in E(G)\}$.
\ENSURE~\\
The new agreements $E^i$ for the next decision period $T$. \\
\STATE $X \gets E^i$
\STATE $x \gets 1$
\WHILE{$X \ne \emptyset$ \AND $x \ne 0$}
\STATE Randomly select $e \gets e_{i,j} \in X$
\STATE $w(e_{i,j}) \gets b^i_j v_c - d^i_j v_d$
\IF{$w(e_{i,j}) < 0$}
\STATE $x = 0$
\ENDIF
\STATE $X \gets X \backslash e$
\ENDWHILE
\IF{$x = 0$}
\STATE $E^i \gets E^i \backslash e$.
\ENDIF
\end{algorithmic}
\end{algorithm}

The network formation rules provide a distributed algorithm for the mobile users to dynamically adjust their D2D agreements in opportunistic offloading, and Theorem~\ref{stability} specifies the conditions to converge to a pairwise stable outcome. However, this algorithm requires a user $i \in \Omega$ to be aware of the distributions of access delays (i.e., parameters $k_j, \lambda_j, j \ne i$) and inter-contact times (i.e., parameters $\alpha_{i,j},\tau_{i,j} , j \ne i$) of other users, which may not be feasible in practice. Also, the conditions under which this algorithm converges to a pairwise stable network are difficult to check, due to the huge computational complexity for calculating $\mathcal{B}$ (e.g., the rule supernetwork $\Upsilon$ contains $2^{\frac{N(N-1)}{2}}$ nodes and $\frac{N(N-1)}{4}  \cdot 2^{\frac{N(N-1)}{2}}$ edges). Therefore, we propose a practical algorithm in Algorithm~\ref{NFA}, in which the users collect historical data to help the decision process. We will show in the simulation section that this practical algorithm can converge to efficient networks in most cases.

Algorithm~\ref{NFA} starts from an initial network structure $G_0$ where opportunistic sharing agreements exist between any two users, i.e., $V(G_0)=\Omega,E(G_0)= \{e_{i,j} | i,j \in \Omega\}$, and each user periodically breaks down the existing links that are not cost-effective so as to increase their individual payoffs. The decision period $T$ is assumed to cover plenty enough times of content deliveries, so that a user can make reliable records between two consecutive decisions. Specifically, a user $i \in \Omega$ records a $1 \times N$ vector $\boldsymbol{d}^{i}$ where $d^i_j \in \mathbb{Z}^+$ represents the cumulative number of D2D communications from user $i$ to user $j$, and a $1 \times N$ vector $\boldsymbol{b}^{i}$ where $b^i_j \in \mathbb{Z}^+$ represents the cumulative number of D2D communications from user $j$ to user $i$. At the end of each decision period $T$, user $i$ gives each link $e_{i,j}$ a weight $w(e_{i,j}) = b^i_j v_c - d^i_j v_d$, which represents the utility that user $j$ brought to user $i$ in this period. If $w(e_{i,j})<0$, then link $e_{i,j}$ is not cost-effective to user $i$ and it should be removed from the network. For each period, we assume user $i$ randomly breaks down one link that is not cost-effective. After a certain number of periods, all the existing links in the network will be cost-effective to both associated end users.

Note that the proposed network formation game and network formation algorithm can be extended to multiple publish/subcribe scenarios. In such cases, the network payoff function should weight and sum up the payoffs from different services. In the network formation process, the users record the D2D communications of each service and jointly consider these data to decide whether an opportunistic sharing agreement is cost-effective. Also, our methods can be extended to data dissemination problems, in which not only the mobility pattern but also the social impact among users should be considered in the payoff function. In such cases, the network payoff function should involve the social impact as well as the access delays of mobile users.

%%%%%%%%%%%%%%%%%%%%%%
\section{Simulations}%
%%%%%%%%%%%%%%%%%%%%%%

In this section, we provide various simulation results to show the performance of Algorithm~\ref{NFA} in terms of convergence, offloading efficiency, and user payoff. Also, we compare the proposed game-theoretical algorithm with a centralized random seeding algorithm, in which the content provider randomly push the data to a set of seed users and then the seed users spread the data to other users via unconditional D2D sharing. In order to specify the parameters for simulation, we assume the parameters of the access delay, $\{k_i\}, \{\lambda_i\}$, are independent variables uniformly distributed in $[k_{min},k_{max}]$ and $[\lambda_{min},\lambda_{max}]$, respectively. As for the contact graph, we assume each user randomly contact at most with $M$ users, and the parameters $\{\alpha_{i,j}\}, \{\tau_{i,j}\}$, are independent variables uniformly distributed in $[\alpha_{min},\alpha_{max}]$ and $[\tau_{min},\tau_{max}]$, respectively.

\subsection{Convergency}

\begin{figure}
\centering
\subfigure[Network connectivity]{
\label{converge-connectivity} %% label for first subfigure
\includegraphics[width=3.0in]{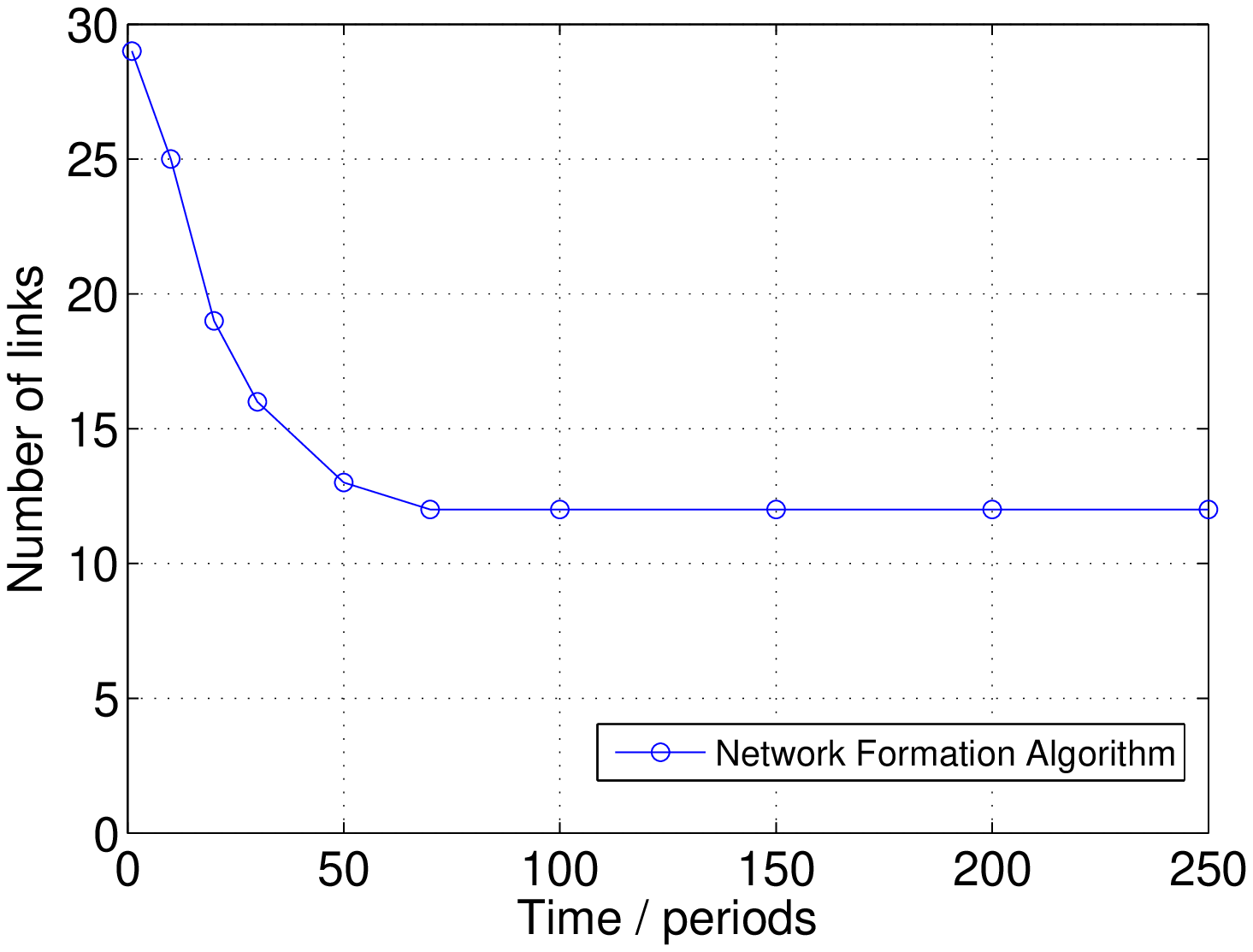}}
\subfigure[Network efficiency]{
\label{converge-efficiency} %% label for second subfigure
\includegraphics[width=3.0in]{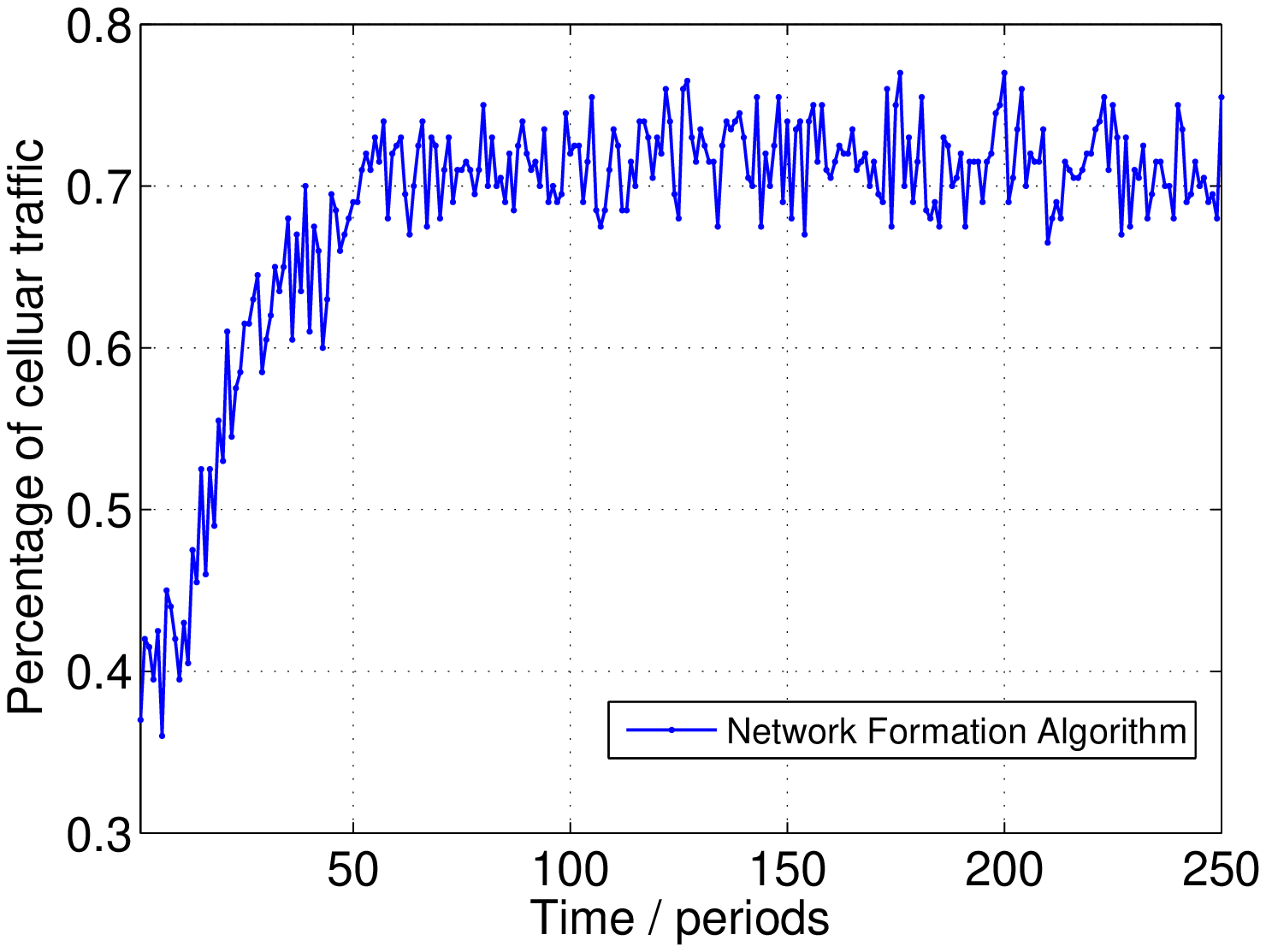}}
\caption{An implementation of the proposed network formation algorithm with respect to time. The parameters are given by $N = 20, v_c / v_d = 4, M=3, k_{min}=2, k_{max}=6, \lambda_{min}=15, \lambda_{max}=45, \alpha_{min}=1,\alpha_{max}=3, \tau_{min}=10, \tau_{max}=15$.}
\label{converge} %% label for entire figure
\end{figure}

In Fig.~\ref{converge}, we show an implementation of the proposed network formation algorithm in the time line, where there are $N=20$ users and each user contacts with at most $M=3$ users. In Fig.~\ref{converge-connectivity}, we show the network connectivity by the number of links. In Fig.~\ref{converge-efficiency}, we show the percentage of cellular traffic compared with the scenario without opportunistic offloading. As we see, the network connectivity drops quickly as the users remove their unwanted links during time periods $t=1$ to $t=70$, and then when $t>70$, the users maintain the cost-effective links and the network becomes stable. The changes in network connectivity highly influence the cellular traffic. As we see in Fig.~\ref{converge-efficiency}, the cellular traffic first increases as the users break down unwanted links, and then become steady as the users converge to a stable network structure where every link is cost-effective to both ends. The fluctuation in Fig.~\ref{converge-efficiency} is due to the statistic characteristics of each implementation in each time period. Note that due to the randomness of the proposed network formation algorithm, the users may converge to different network structures even for the same implementation.

\begin{figure}
\centering
\subfigure[$v_c/v_d=4$]{
\label{convergetime-N} %% label for first subfigure
\includegraphics[width=3.0in]{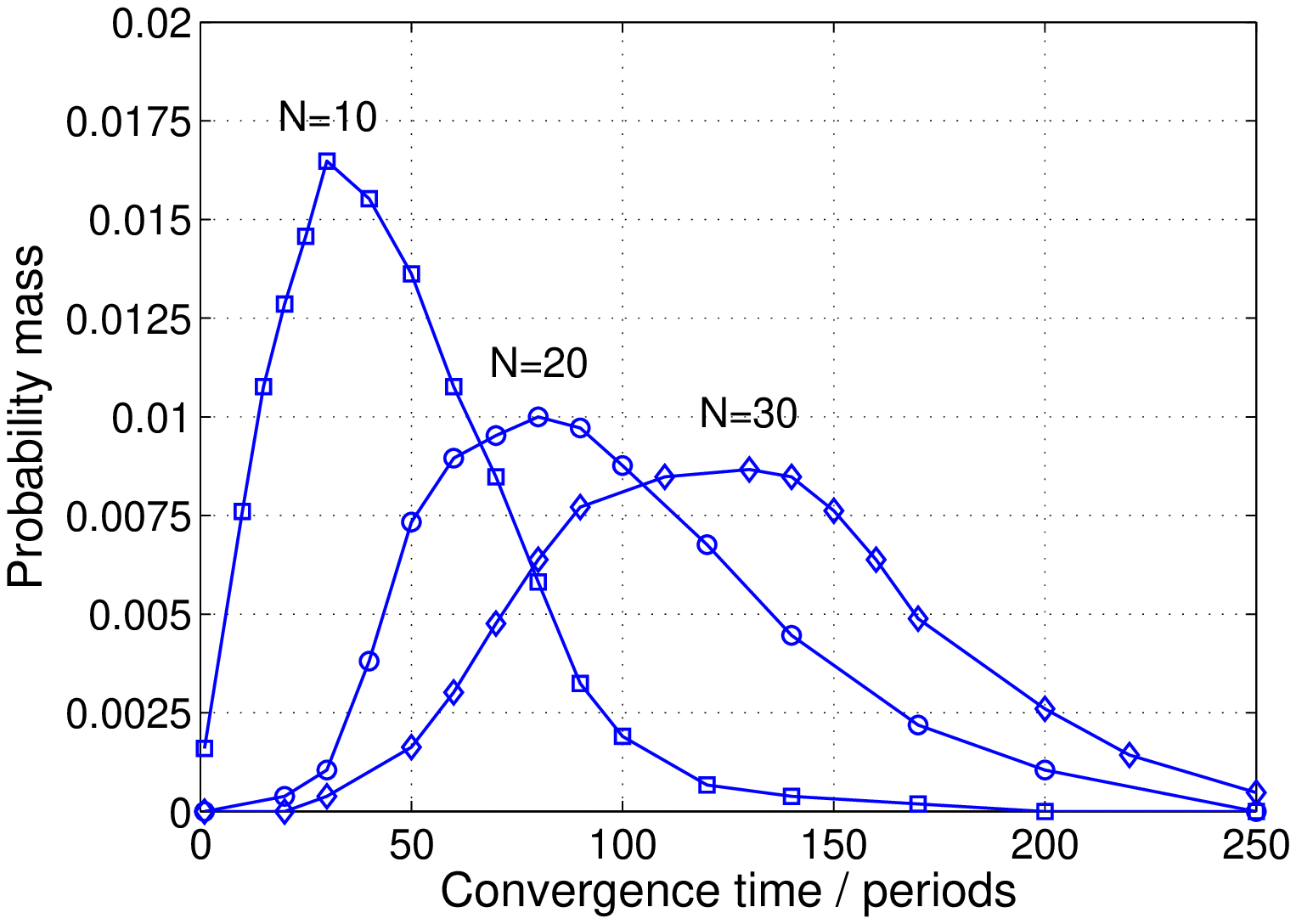}}
\subfigure[$N=20$]{
\label{convergetime-rate} %% label for second subfigure
\includegraphics[width=3.0in]{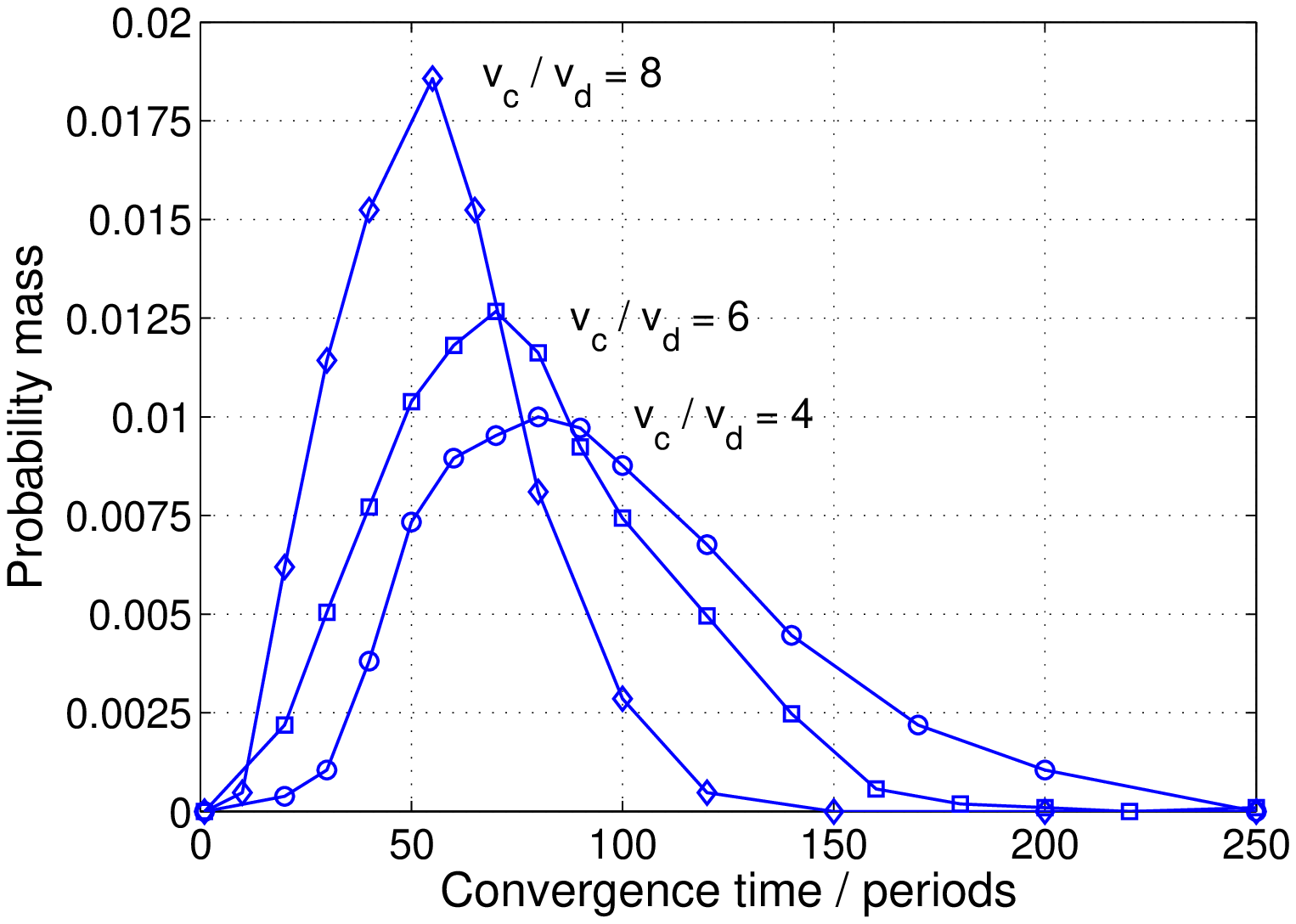}}
\caption{Probability mass function of convergence time of the proposed network formation algorithm. The parameters are given by $M=3, k_{min}=2, k_{max}=6, \lambda_{min}=15, \lambda_{max}=45, \alpha_{min}=1,\alpha_{max}=3, \tau_{min}=10, \tau_{max}=15$.}
\label{convergetime} %% label for entire figure
\end{figure}

In Fig.~\ref{convergetime}, we show the distribution of convergence time of the proposed network formation algorithm with different user numbers $N = 10, 20, 30$ and different cost ratios $v_c / v_d = 4$. In Fig.~\ref{convergetime-N}, we see that, as the network scale increases, the users require more time to converge to a stable network, and the convergence time disperses in a larger range. Note that practical users tend to be organized by multiple small groups rather a large low-connected graph~\cite{ZNKT-2007, KLV-2010}. Thus, the proposed algorithm can be applied to large-scale networks, and the convergence time is decided by the largest group in the network. In Fig.~\ref{convergetime-rate}, we see that, as the cost ratio increases, which implies that the cost of D2D transmissions relatively decreases, the users are more willing to maintain sharing agreements with other users, and thus, the convergence time decreases and concentrates on a smaller range.

\subsection{Efficiency}

\begin{figure}[!t]
\centering
\includegraphics[width=4.2in]{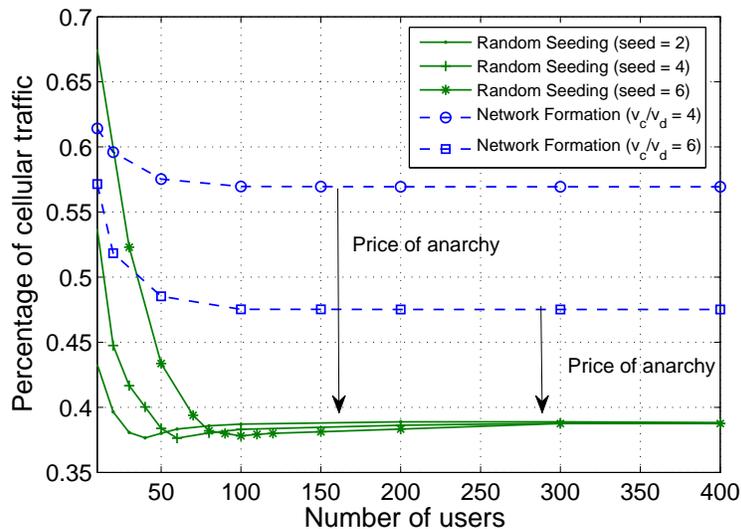}
\caption{Efficiency of the proposed network formation algorithm and referential random seeding method with respect to the number users $N$. The parameter are given by $M=3, k_{min}=2, k_{max}=6, \lambda_{min}=15, \lambda_{max}=45, \alpha_{min}=1,\alpha_{max}=3, \tau_{min}=10, \tau_{max}=15$.}
\label{efficiency-N}
\end{figure}

In Fig.~\ref{efficiency-N}, we show the percentage of cellular traffic with respect to the number of users $N$ for both the proposed network formation algorithm and the referential random seeding algorithm. As we can see, the proposed algorithm offloads $43\%$ cellular traffic when $v_c/v_d = 4$, and $52\%$ cellular traffic when $v_c/v_d = 6$, while the random seeding algorithm offloads $61\%$ cellular traffic. The gap between these two algorithms measures how the network efficiency degrades due to the selfish behavior of the mobile users (i.e., breaking down non-profitable links), which is often referred to as \emph{the price of anarchy} in game theory~\cite{PR-2013}. Therefore, the selfishness of mobile users can cost a high performance loss in practical scenarios, and the loss is highly influenced by the cost ratio $v_c/v_d$.

For the random seeding algorithm, we see that its performance highly depends on the number of seed users, as in Fig.~\ref{efficiency-N}. For any number of users $N$, there exists an optimal number of seed users. If the seed users exceed the optimal number, some seed users that may have the chance to be delivered via D2D transmissions, have to download the content via cellular transmissions, which unnecessarily increases the cellular traffic. If the seed users are below the optimal number, some of the non-seed users will not be able to receive the content via D2D transmissions before their access delays, and thus, they have to download the content from the base station and increase the cellular traffic. The envelope of all the random seeding curves provides the optimal performance of random seeding algorithm. As we see, when the network scale becomes extremely large, all the curves converge to the case where there is no seed users, since the seed number is negligible compared with the entire network. Note that even when the seed number is zero, the users still get the content after their access delay. Thus, D2D transmissions still exist among the users with different access delays and the network can still offload cellular traffic.

For the proposed network formation algorithm, we see that its performance highly depends on the cost ratio $v_c/v_d$, as in Fig.~\ref{efficiency-N}. For any cost ratio $v_c/v_d$, we can see that its cellular traffic decreases as the number of users $N$ increases, which implies that the offloading efficiency increases with the network scale. When $N$ is small, the network is relatively highly connected, and the selfish behavior of the users with short access delay, i.e., removing the non-profitable links, may influence a large population. When $N$ is large, the network is relatively sparse, and the influence of those selfish behaviors are restricted to a smaller population. Therefore, the performance of the proposed algorithm improves as the number of uses $N$ increases. When $N$ is extremely large, the network is extremely sparse where the selfish behaviors only influence the neighbor users. Since the neighbors of each user is a fixed number $M$, the network performance converges when $N$ is extremely large.

\begin{figure}[!t]
\centering
\includegraphics[width=4.2in]{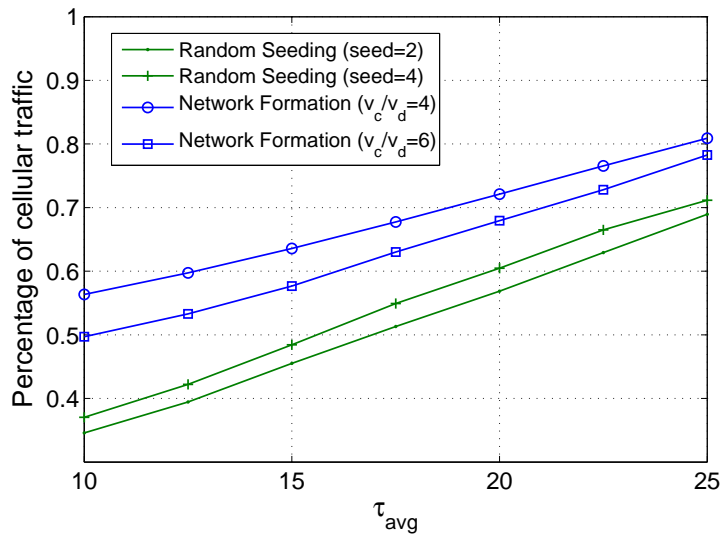}
\caption{Efficiency of the proposed network formation algorithm and referential random seeding method with respect to $\tau_{avg} = (\tau_{min}+\tau_{max})/2$. The parameter are given by $N=20, M=3, k_{min}=2, k_{max}=6, \lambda_{min}=15, \lambda_{max}=45, \alpha_{min}=1,\alpha_{max}=3$.}
\label{efficiency-Tau}
\end{figure}

In Fig.~\ref{efficiency-Tau}, we show the percentage of cellular traffic with respect to $\tau_{avg} = (\tau_{min}+\tau_{max})/2$ for both the proposed network formation algorithm and the referential random seeding algorithm. Here, $\tau_{avg}$ represents the average value of the minimum inter-contact time. As we see, the percentage of cellular traffics increases linearly with $\tau_{avg}$ for both algorithms. This figure indicates that for both the proposed algorithm and the random seeding algorithm, the network efficiency will degrade gradually as the user mobility decreases. Also, we see that the gap between these two algorithms decreases as $\tau_{avg}$ increases, which implies that the price of anarchy becomes smaller as the users have higher mobility.

\begin{figure}[!t]
\centering
\includegraphics[width=4.2in]{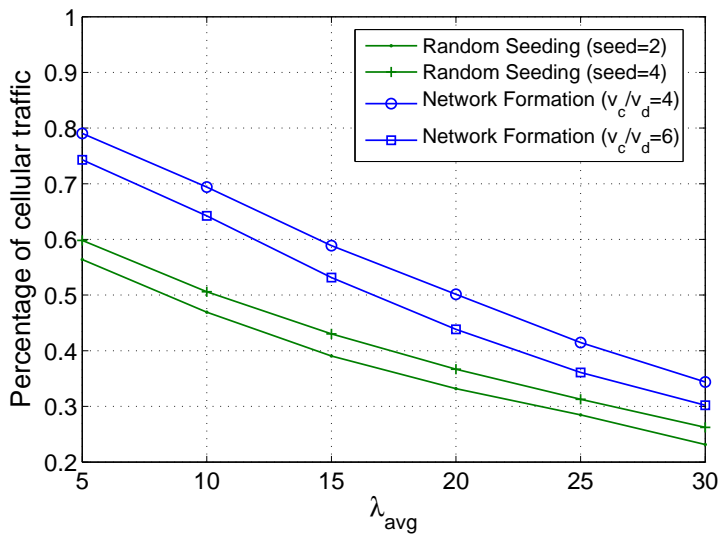}
\caption{Efficiency of the proposed network formation algorithm and referential random seeding method with respect to $\lambda_{avg} = (\lambda_{min}+\lambda_{max})/2$. The parameter are given by $N=20, M=3, k_{min}=2, k_{max}=6, \alpha_{min}=1, \alpha_{max}=3, \tau_{min}=10, \tau_{max}=15$.}
\label{efficiency-Lambda}
\end{figure}

In Fig.~\ref{efficiency-Lambda}, we show the percentage of cellular traffic with respect to $\lambda_{avg} = (\lambda_{min}+\lambda_{max})/2$ for both the proposed network formation algorithm and the referential random seeding algorithm. Here, $\lambda_{avg}$ represents the rough expected value of the user's access delay. As we see, the percentage of cellular traffics decreases linearly with $\lambda_{avg}$ for both algorithms. This figure indicates that for both the proposed algorithm and the random seeding algorithm, the network efficiency will improve gradually as the users tolerate larger delays. Also, we see that the gap between these two algorithms decreases as $\lambda_{avg}$ increases, which implies that the price of anarchy becomes smaller as the users can tolerate larger delays.

\subsection{User Payoff}

\begin{figure}
\centering
\subfigure[Random seeding]{
\label{payoff-a} %% label for first subfigure
\includegraphics[width=3.0in]{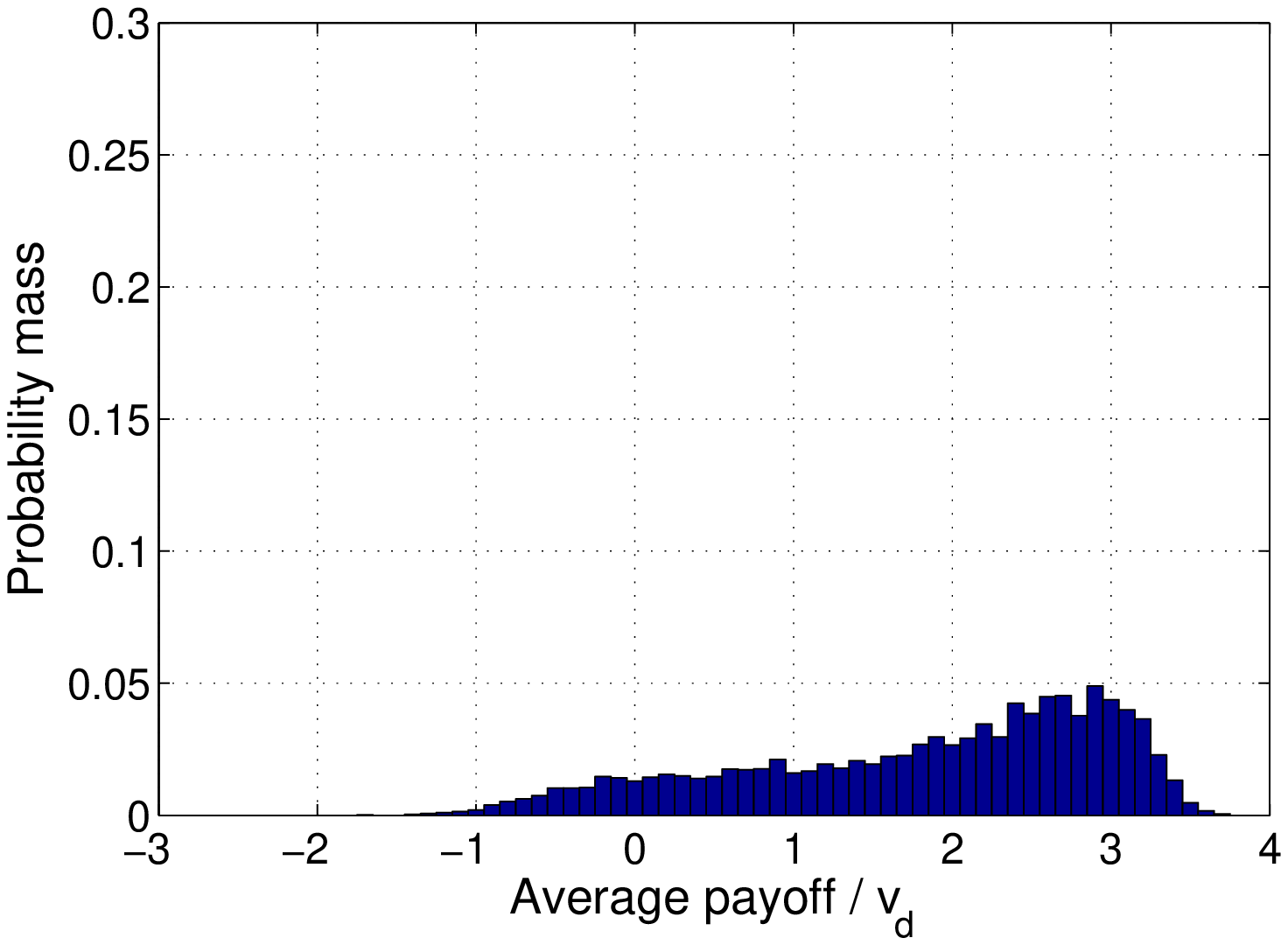}}
\subfigure[Network formation]{
\label{payoff-b} %% label for second subfigure
\includegraphics[width=3.0in]{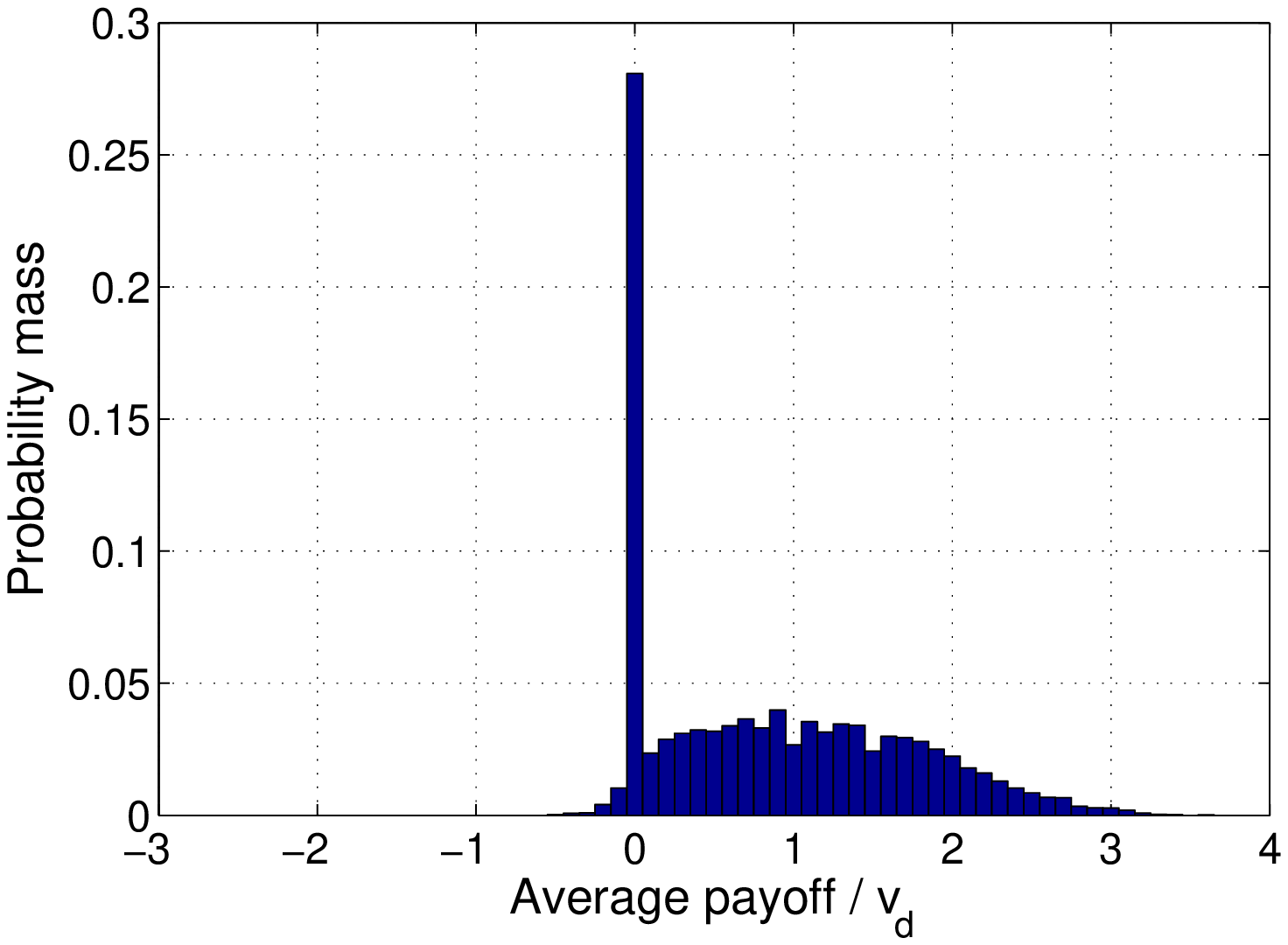}}
\caption{The average payoff distribution of both the proposed network formation algorithm and the referential random seeding algorithm. The parameters are given by $N = 20, v_c / v_d = 4, M=3, k_{min}=2, k_{max}=6, \lambda_{min}=15, \lambda_{max}=45, \alpha_{min}=1,\alpha_{max}=3, \tau_{min}=10, \tau_{max}=15$.}
\label{payoff-dis} %% label for entire figure
\end{figure}

In Fig.~\ref{payoff-dis}, we show the distribution of user payoff for both the proposed network formation algorithm and the referential random seeding algorithm. In Fig.~\ref{payoff-a}, we see that the random seeding algorithm may lead to negative payoffs. While in Fig.~\ref{payoff-b}, we see that the proposed network formation algorithm can guarantee that all users have a non-negative payoff, since only cost-effective links are maintained in the network formation process. Note that a large portion of users in the network formation algorithm do not have any sharing agreements with other users, and thus, receives a zero payoff. Also, there still exists a very small portion of negative payoffs in Fig.~\ref{payoff-b}, and it is due to the statistic characteristics of each implementation.

%%%%%%%%%%%%%%%%%%%%%%%%%%%%%%%%%%%%%%%%%%%%%%%
\section{Conclusions}%
%%%%%%%%%%%%%%%%%%%%%%%%%%%%%%%%%%%%%%%%%%%%%%%

In this paper, we have studied the social data offloading of selfish users, for which we have proposed a network formation game to capture the characteristics of selfish behaviors, and a data-based network formation algorithm to guarantee positive payoffs of each user. Our analysis have shown that the users can be guaranteed to converge to a pairwise stable network under some conditions. And simulation results have shown that the performance of such pairwise stable networks can be highly degraded, compared with the ideal scenarios of selfishless users. The performance gap between selfish users and selfishless users becomes smaller as the cost ratio of cellular and D2D transmissions increases, as the users have higher mobility, or as the users can tolerate larger delays.

\end{document}